\documentclass[11pt,letterpaper]{article}
\usepackage[tbtags]{amsmath}
\usepackage{amssymb}
\usepackage{amsthm}
\usepackage{mathrsfs}
\usepackage{xcolor}
\usepackage{enumitem}
\usepackage{mathtools}
\usepackage{multirow}
\usepackage{hyperref}
\usepackage{rotating}
\usepackage[all,2cell,dvips]{xy} \UseAllTwocells \SilentMatrices

\newcommand{\cat}{\circ}
\DeclareMathOperator{\Inf}{\mathit{Inf}}

{\bfseries}{\itshape}
{\bfseries}{\itshape}
{\bfseries}{\itshape}
\newtheorem{definition}{Definition}{\bfseries}{\itshape}
{\bfseries}{\itshape}
\newtheorem{example}{Example}{\bfseries}{\itshape}
\newtheorem{theorem}{Theorem}{\bfseries}{\itshape}
\newtheorem{lemma}{Lemma}{\bfseries}{\itshape}

\newenvironment{varthm}[1][Theorem]{\begin{trivlist}
\item[\hskip \labelsep {\bfseries #1}]}{\end{trivlist}}

\newcommand{\bbN}{\mathbb{N}}

\newcommand{\calA}{\mathcal{A}}
\newcommand{\calB}{\mathcal{B}}
\newcommand{\calC}{\mathcal{C}}
\newcommand{\calD}{\mathcal{D}}

\newcommand{\calF}{\mathcal{F}}

\newcommand{\calS}{\mathcal{S}}

\newcommand{\calW}{\mathcal{W}}

\newcommand{\scrG}{\mathscr{G}}

\newcommand{\scrL}{\mathscr{L}}

\begin{document}

\title{A Tight Lower Bound for Streett Complementation}

\author{Yang Cai\\
MIT CSAIL \\
The Stata Center, 32-G696 \\
Cambridge, MA 02139 USA \\
ycai@csail.mit.edu \\
\and
Ting Zhang \\
Iowa State University \\
226 Atanasoff Hall \\
Ames, IA 50011 USA\\
tingz@iastate.edu
}

\date{}
\maketitle
\thispagestyle{empty}

\begin{abstract}
Finite automata on infinite words ($\omega$-automata) proved to be a powerful weapon for modeling and reasoning infinite behaviors of reactive systems. Complementation of $\omega$-automata is crucial in many of these applications. But the problem is non-trivial; even after extensive study during the past two decades, we still have an important type of $\omega$-automata, namely Streett automata, for which the gap between the current best lower bound $2^{\Omega(n \lg nk)}$ and upper bound $2^{\Omega(nk \lg nk)}$ is substantial, for the Streett index size $k$ can be exponential in the number of states $n$. In~\cite{CZ11b} we showed a construction for complementing Streett automata with the upper bound $2^{O(n \lg n+nk \lg k)}$ for $k = O(n)$ and $2^{O(n^{2} \lg n)}$ for $k=\omega(n)$. In this paper we establish a matching lower bound $2^{\Omega(n \lg n+nk \lg k)}$ for $k = O(n)$ and $2^{\Omega(n^{2} \lg n)}$ for $k = \omega(n)$, and therefore showing that the construction is asymptotically optimal with respect to the $2^{\Theta(\cdot)}$ notation.
\end{abstract}

\thispagestyle{empty}

\section{Introduction}
\label{sec:introduction}
Complementation is a fundamental notion in automata theory. Given an automaton $\calA$, the complementation problem asks to find an automaton $\calB$ that accepts exactly all words that $\calA$ does not accept. Complementation connects automata theory with mathematical logic due to the natural correspondence between language complementation and logical negation, and hence plays a pivotal role in solving many decision and definability problems in mathematical logic.

A fundamental connection between automata theory and the monadic second order logics was demonstrated by B\"{u}chi~\cite{Buc60}, who started the theory of finite automata on infinite words ($\omega$-automata)~\cite{Buc66}. The original $\omega$-automata are now referred to as B\"{u}chi automata and B\"{u}chi complementation was a key to establish that the class of $\omega$-regular languages (sets of $\omega$-words generated by product $\cat$, union $\cup$, star ${}^{*}$ and limit ${}^{\omega}$) is closed under complementation~\cite{Buc66}.

B\"{u}chi's discovery also has profound repercussions in applied logics. Since the '80s, with increasing demand of reasoning infinite computations of reactive and concurrent systems, $\omega$-automata have been acknowledged as unifying representation for \emph{programs} as well as for \emph{specifications}~\cite{VW86}. Complementation of $\omega$-automata is crucial in many of these applications.

But complementation of $\omega$-automata is non-trivial. Only after extensive studies in the past two decades ~\cite{SVW87,Mic88,Saf88,FKV06,Yan06,Sch09} (also see survey~\cite{Var07}), do we have a good understanding of the complexity of B\"{u}chi complementation. But a question about a very important type of $\omega$-automata remains unanswered, namely the complexity of Streett complementation, where the gap between the current lower bound and upper bound is substantial. Streett automata are ones of a kind, because Streett acceptance conditions naturally encode \emph{strong fairness} that infinitely many requests are responded infinitely often, a necessary requirement for meaningful computations~\cite{FK84,Fra86}.

\paragraph{Related Work.}
Obtaining nontrivial lower bounds has been difficult. The first nontrivial lower bound for B\"{u}chi complementation is $n!\approx (0.36n)^{n}$, obtained by Michel~\cite{Mic88,Lod99}. In 2006, combining ranking with \emph{full automaton} technique, Yan improved the lower bound of B\"{u}chi complementation to $\Omega(L(n))$~\cite{Yan06}, which now is matched tightly by the upper bound $O(n^{2}(L(n))$~\cite{Sch09}, where $L(n)\approx(0.76n)^{n}$. Also established in~\cite{Yan06} was a $(\Omega(nk))^{n}=2^{\Omega(n\lg nk)}$ tight lower bound (where $k$ is the number of B\"{u}chi indices) for generalized B\"{u}chi complementation, which also applies to Streett complementation because generalized B\"{u}chi automata are a subclass of Streett automata. In~\cite{CZL09}, we proved a tight lower bound $2^{\Omega(nk \lg n)}$ for Rabin complementation (where Rabin index size $k$ can be as large as $2^{n-\epsilon}$ for any arbitrary but fixed $\epsilon > 0$). Several constructions for Streett complementation exist~\cite{SV89,Kla91,Saf92,KV05a,Pit06}, but all involve at least $2^{O(nk\lg nk)}$ state blow-up, which is significantly higher than the current best lower bound $2^{\Omega(n\lg nk)}$, since the Streett index size $k$ can reach $2^{n}$. Determining the complexity of Streett complementation has been posed as an open problem since the late '80s~\cite{SV89,KV05a,Yan06,Var07}. In~\cite{CZ11b} we showed a construction for Streett complementation with the upper bound $2^{O(n \lg n+nk \lg k)}$ for $k = O(n)$ and $2^{O(n^{2} \lg n)}$ for $k=\omega(n)$. In this paper we establish a matching lower bound $2^{\Omega(n \lg n+nk \lg k)}$ for $k = O(n)$ and $2^{\Omega(n^{2} \lg n)}$ for $k = \omega(n)$, and therefore showing that the construction in~\cite{CZ11b} is essentially optimal at the granularity of $2^{\Theta(\cdot)}$. This lower bound is obtained by applying two techniques: \emph{fooling set} and \emph{full automaton}.

\paragraph{Fooling Set.}
The fooling set technique is a classic way of obtaining lower bounds on nondeterministic finite automata on finite words (NFA). Let $\Sigma$ be an alphabet and $\scrL \subseteq \Sigma^{*}$ a regular language. A set of pairs $P=\{(x_{i}, y_{i}) \mid x_{i}, y_{i} \in \Sigma^{*}, 1 \le i \le n \}$ is called a \emph{fooling set} for $\scrL$, if $x_{i}y_{i} \in \scrL$ for $1 \le i \le n$ and $x_{i}y_{j} \not \in \scrL$ for $1 \le i,j \le n$ and $i \not = j$. If $\scrL$ has a fooling set $P$, then any NFA accepting $\scrL$ has at least $|P|$ states~\cite{GS96}. The purpose of a fooling set is to identify runs with dual properties (called fooling runs): fragments of accepting runs of $\scrL$, when pieced together in certain ways, induce non-accepting runs. By an argument in the style of Pumping Lemma, a small automaton would not be able to distinguish how it arrives at a state, and hence it cannot differentiate between some accepting runs and some non-accepting ones.

In the setting of $\omega$-automata, a similar technique exists, which we refer to as Michel's scheme~\cite{Mic88}. A set $P=\{x_{i} \in \Sigma^{*} \mid 1 \le i \le n \}$ is called a \emph{fooling set} for $\scrL$, if $(x_{i})^{\omega} \in \scrL$ for $1 \le i \le n$ and $((x_{i})^{+}(y_{j})^{+})^{\omega} \subseteq \overline{\scrL}$ for $1 \le i,j \le n$ and $i \not = j$~\cite{Mic88,Lod99}.

\paragraph{Full Automaton.}
Sakoda and Sipser introduced the \emph{full automaton} technique~\cite{SS78} (the name was first coined in~\cite{Yan06}) and used it to obtain several completeness and lower bound results on transformations involving $2$-way finite automata~\cite{SS78}. In particular, they proved a classic result of automata theory: the lower bound of complementing an NFA with $n$ states is $2^{n}$.

To establish lower bounds for complementation, one starts with designing a class of automata $\calA_{n}$ and then a class of words $\calW_{n}$ such that $\calW_{n}$ are not contained in $\scrL(\calA_{n})$. Next one shows that runs of purported complementary automata $\calC_{n}$ on $\calW_{n}$ exhibit dual properties by application of the fooling set technique. However, some fooling runs can only be generated by long and sophisticated words, which are very difficult to be ``guessed'' right from the beginning. The ingenuity of the full automaton technique is to remove two levels of indirections: since the ultimate goal is to construct fooling runs, why should not one start with runs directly, and build $\calW_{n}$ and $\calA_{n}$ later?

Without a priori constraints imposed from $\calA_{n}$ or $\calW_{n}$ (they do not exist yet), full automata operate on all possible runs; for a full automaton of $n$ states, every possible unit transition graph (bipartite graph with $2n$ vertices) is identified with a letter, and words are nothing but potential run graphs. Removing the two levels of indirections proved to be powerful. By this technique, the $2^{n}$ lower bound proof for complementing NFA was surprisingly short and easy to understand~\cite{SS78} (a fooling set method was implicit in the proof).

We should note that full automata operate on large alphabets whose size grows exponentially with the state size, but this does not essentially limit its application to automata on conventional alphabets. By an encoding trick, a large alphabet can be mapped to a small alphabet with no compromise to lower bound results~\cite{Sip79,Yan06,CZL09}.

\paragraph{Ranking.}
For $\omega$-automata, the power of fooling set and full automaton technique was further enhanced by the use of rankings on run graphs~\cite{Yan06,CZL09}. Since first introduced in~\cite{Kla91}, rankings have been shown to a powerful tool to represent properties of run graphs; complementation constructions for various types of $\omega$-automata were obtained by discovering respective rankings that precisely characterize those run graphs that contain no accepting path (with respect to source automata)~\cite{KV01,KV04,KV05a,FKV06,Kup06}. With the help of rankings, constructing a fooling set amounts to designing certain type of rankings. In fact, as shown below, an explicit description of a fooling set might be very hard to find, but the essential properties the fooling set induce can be concisely represented by certain type of rankings.

\paragraph{Our Results.}
In this paper we establish a lower bound $L(n,k)$ for Streett complementation: $2^{\Omega(n\lg n + k n\lg k)}$ for $k=O(n)$ and $2^{\Omega(n^{2} \lg n)}$ for $k=\omega(n)$, which matches the upper bound obtained in~\cite{CZ11b}. This lower bound applies to all Streett complementation constructions that output union-closed automata (see Section~\ref{sec:preliminaries}), which include B\"{u}chi, generalized B\"{u}chi and Streett automata. This bound considerably improves the current best bound $2^{\Omega(n \lg nk)}$~\cite{Yan06}, especially in the case $k = \Theta(n)$.

Determinization is another fundamental concept in automata theory and it is closely related to complementation. A deterministic $T$-automaton can be easily complemented by switching from $T$-acceptance condition to the dual co-$T$ condition (e.g., Streett vs. Rabin). Therefore, the lower bound $L(n,k)$ also applies to Streett determinization if the output automata are the dual of union-closed automata. In particular, no construction for Streett determinization can output Rabin automata with state size asymptotically less than $L(n,k)$.

We can get a slightly weaker result for constructions that output Rabin automata (which are not union-closed): no construction for Streett complementation can output Rabin automata with state size $n' \le L(n,k)$ and index size $k' = O(n')$, due to the fact that a Rabin automaton with state $n'$ and index size $k'$ can be translated to an equivalent B\"{u}chi automaton with $O(n'k')$ states. For the same reason, no construction for Streett determinization can output Streett automata with state size $n' \le L(n,k)$ and index size $k' = O(n')$.

Even with the fooling set and full automaton techniques and the assistance of rankings, a difficulty remains: in the setting of Streett complementation, how large can a fooling set for a complementary automaton be? The challenge is two-fold. One is to implant potentially contradictory properties in each member of a fooling set so that complementary run graphs can be obtained by certain combinations of those members. The other is to avoid correlations between members of a fooling set so that each member has to be memorized by a distinct state in a purported complementary automaton. By exploiting the nature of Streett acceptance conditions, our fooling set is obtained via a type of multi-dimensional rankings, called $Q$-rankings, and members in the fooling set are called $Q$-words. To simultaneously accommodate potentially contradictory properties in multi-dimension requires handling nontrivial subtleties. We shall continue this discussion in Section~\ref{sec:lower-bound} after presenting the definition of $Q$-rankings.

\paragraph{Paper Organization.}
Section~\ref{sec:preliminaries} presents notations and basic terminology in automata theory. Section~\ref{sec:lower-bound} introduces full Streett automata, $Q$-rankings and $Q$-words, and use them to establish the lower bound. Section~\ref{sec:conclusion} concludes with a discussion. Technical proofs are omitted from the main text, but they can be found in the appendix.
%Section~\ref{sec:construction} presents the proofs for the existence of $Q$-words.

\section{Preliminaries}
\label{sec:preliminaries}

\paragraph{Basic Notations.}
Let $\bbN$ be the set of natural numbers. We write $[i..j]$ for $\{ k \in \bbN \, \mid \, i \le k \le j\}$, $[i..j)$ for $[i..j-1]$, $[n]$ for $[0..n)$. For an infinite sequence $\varrho$, we use $\varrho(i)$ to denote the $i$-th component for $i \in \bbN$, $\varrho[i..j]$ (resp. $\varrho[i..j)$) to denote the subsequence of $\varrho$ from position $i$ to position $j$ (resp. $j-1$). Similar notations for finite sequences and we use $|\varrho|$ to denote the length of $\varrho$. We assume readers are familiar with notations in language theory, such as $\alpha \cat \alpha'$, $\alpha^{*}$, $\alpha^{+}$ and $\alpha^{\omega}$ where $\alpha$ and $\alpha'$ are sequences and $\alpha$ is finite, and similar ones such as $S \cat S'$, $S^{*}$, $S^{+}$ and $S^{\omega}$ where $S$ is a set of finite sequences and $S'$ is a set of sequences.

\paragraph{Automata and Runs.}
A finite (nondeterministic) automaton on infinite words ($\omega$-automaton) is a $5$-tuple $\calA=\langle\Sigma, S, Q, \Delta, \calF\rangle$, where $\Sigma$ is an alphabet, $S$ is a finite set of states, $Q \subseteq S$ is a set of initial states, $\Delta \subseteq S \times \Sigma \times S$ is a transition relation, and $\calF$ is an acceptance condition.

An infinite word ($\omega$-words) over $\Sigma$ is an infinite sequence of letters in $\Sigma$. A \emph{run} $\varrho$ of $\calA$ over an $\omega$-word $w$ is an infinite sequence of states in $S$ such that $\varrho (0)\in Q$ and, $\langle\varrho(i),w(i),\varrho(i\!+\!1)\rangle \in \Delta$ for $i \in \bbN$. Finite runs are defined similarly. Let $\Inf(\varrho)$ the set of states that occur infinitely many times in $\varrho$. An automaton accepts $w$ if there exists a run $\varrho$ over $w$ that satisfies $\calF$, which usually is defined as a predicate on $\Inf(\varrho)$. We use $\scrL(\calA)$ to denote the set of $\omega$-words accepted by $\calA$ and $\overline{\scrL(\calA)}$ the complement of $\scrL(\calA)$.

\paragraph{Acceptance Conditions and Automata Types.}
$\omega$-automata are classified according their acceptance conditions. Below we list three types of $\omega$-automata relevant to this paper. Let $F$ be a subset of $Q$ and $G, B$ two functions $I \to 2^{Q}$ where $I=[1..k]$ is called the \emph{index set}.
\begin{itemize}
\item \emph{B\"{u}chi}: $\langle F \rangle$: $\Inf(\varrho) \cap F\neq\emptyset$.
\item \emph{Streett}: $\langle G,B\rangle_{I}$: $\forall i \in I$, $\Inf(\varrho) \cap G(i)\neq\emptyset \to \Inf(\varrho)\cap B(i)\neq\emptyset$.
\item \emph{Rabin}: $[G,B]_{I}$: $\exists i \in I$, $\Inf(\varrho)\cap G(i)\neq\emptyset \wedge \Inf(\varrho)\cap B(i)=\emptyset$.
\end{itemize}
Note that Streett and Rabin are dual to each other. An automaton $\calA$ is called \emph{union-closed} if when two runs $\varrho$ and $\varrho'$ are accepting, so is any run $\varrho''$ if $\Inf(\varrho'') = \Inf(\varrho) \cup \Inf(\varrho')$. It is easy to verify that both B\"{u}chi and Streett automata are union-closed while Rabin automata are not. Let $J \subseteq I$. We use $\langle G,B\rangle_{J}$ to denote the Streett condition with respect to only indices in $J$. When $J$ is a singleton, say $J=\{j\}$, we simply write $\langle G(j),B(j)\rangle$ for $\langle G,B\rangle_{J}$. We can assume that $B$ is injective and the index size $k$ is bound by $2^{n}$, because if $B(i)=B(i')$ for two different $i,i' \in I$, then we can shrink the index set $I$ by replacing $\langle G,B\rangle_{\{i,i'\}}$ by $\langle G(i)\cup G(i'), B(i) \rangle$. The same convention and assumption are used for Rabin condition.

\paragraph{$\Delta$-Graphs.}
A \emph{$\Delta$-graph} (run graph) of an $\omega$-word $w$ under $\calA$ is a directed graph $\scrG_{w}=(V,E)$ where $V=S \times \bbN$ and $E=\{\langle\langle s,l\rangle,\langle s',l+1\rangle\rangle \in V \times V \, \mid \, s, s' \in S,\ l \in \bbN, \langle s,w(l),s'\rangle\in\Delta \, \}$. By the \emph{$l$-th level}, we mean the vertex set $S \times \{l\}$. Let $S=\{s_{0}, \ldots, s_{n-1}\}$. By \emph{$s_{l}$-track} we mean the vertex set $\{s_{l}\} \times \bbN$. For a subset $X$ of $S$, we call a vertex $\langle s, l\rangle$ an $X$-vertex if $s \in X$. We simply use $s$ for $\langle s, l \rangle$ when the index is irrelevant.

A \emph{$\Delta$-graph} $\scrG_{w}$ of a finite word $w$ is defined similarly. By $|\scrG_{w}|$ we denote the length of $\scrG_{w}$, which is the same as $|w|$. $\scrG_{\sigma}$ for $\sigma \in \Sigma$ is called a \emph{unit} $\Delta$-graph. A path in $\scrG_{w}$ is called a \emph{full path} if the path goes from level $0$ to level $|\scrG_{w}|$. By $\scrG_{w} \cat \scrG_{w'}$, we mean the concatenation of $\scrG_{w}$ and $\scrG_{w'}$, which is the graph obtained by merging the last level of $\scrG_{w}$ with the first level of $\scrG_{w'}$. Note that $\scrG_{w} \cat \scrG_{w'} = \scrG_{w \cat w'}$.

Let $w$ be a finite word. For $l, l' \in \bbN$, $s, s' \in S$ we write $\langle s, l\rangle \xrightarrow{w} \langle s', l'\rangle$ to mean that there exists a run $\varrho$ of $\calA$ such that $\varrho[l..l']$, the subsequence $\varrho(l)\varrho(l+1) \cdots \varrho(l')$ of $\varrho$, is a finite run of $\calA$ from $s$ to $s'$ over $w$. We simply write $s \xrightarrow{w} s'$, when omitting level indices causes no confusion.

\paragraph{Full Automata.}

A full automaton $\langle\Sigma, S, Q,\Delta,\calF\rangle$ is a finite automaton with the following conditions: $\Sigma = 2^{S \times S}$, $\Delta \subseteq S \times 2^{S \times S} \times S$, and for all $s, s' \in S$, $\sigma \in \Sigma$, $\langle s, \sigma, s' \rangle \in \Delta$ if and only if $\langle s, s' \rangle \in \sigma$~\cite{SS78,Yan06,CZL09}. For full automata, the alphabet $\Sigma$ and the transition relation $\Delta$ are completely determined by $S$. As stated in the introduction, the essence of full automaton technique is to use run graphs as free as possible, without worrying which word generates which run graph. Let the functional version of $\Delta$ be $\delta: \Sigma \to 2^{S \times S}$, where for every $s,s' \in S$ and every $\sigma \in \Sigma$, $\langle s, s' \rangle \in \delta(\sigma)$ if and only if $\langle s, \sigma, s' \rangle \in \Delta$. The function $\delta$ maps a letter $\sigma$ to a unit $\Delta$-graph $\scrG_{\sigma}$, which represents the complete behavior of $\calA$ over $\sigma$ (technically speaking, $\scrG_{\sigma}$, with index dropped, is the graph of $\delta(\sigma)$). In the setting of full automata, $\delta$ is simply the identity function on $2^{S \times S}$. Words and run graphs are essentially the same thing. From now on we use the two terms interchangeably. For example, for a word $w$, $s \xrightarrow{w} s'$ is equivalent to say that a full path in $\scrG_{w}$ goes from $s$ to $s'$. 
\section{Lower Bound}
\label{sec:lower-bound}

In this section we define full Streett automata, and related $Q$-rankings and $Q$-words, and use them to establish the lower bound. From now on, we reserve $n$ and $k$, respectively, for the effective state size and index size in our construction (except in Theorem~\ref{thm:lower-bound} and Section~\ref{sec:conclusion} where $n$ and $k$, respectively, mean the state size and index size of a complementation instance). All related notions are in fact parameterized with $n$ and $k$, but we do not list them explicitly unless required for clarity. Let $I$ be $[1..k]$. We first describe the plan of proof.

For each $k, n >0$, we define a full Streett automaton $\calS=(\Sigma,S,Q,\Delta,\calF)$ and a set of $Q$-rankings $f: Q \to [1..n] \times I^{k}$. For each $Q$-ranking $f$, we define a finite $\Delta$-graph $\scrG_{f}$, called a $Q$-word. We then show that for each $f$, $(\scrG_{f})^{\omega} \not \in \scrL(\calS)$, yet $((\scrG_{f})^{+}(\scrG_{f'})^{+})^{\omega} \subseteq \scrL(\calS)$ for every distinct pair of $Q$-rankings $f$ and $f'$, that is, $Q$-words constitute a fooling set for $\overline{\scrL(\calS)}$. Using Michel's scheme~\cite{Mic88,Lod99,Yan06}, we show that if a union-closed automaton $\calC$ complements $\calS$, then its state size is no less than the number of $Q$-rankings, because otherwise we can ``weave'' the runs of $(\scrG_{f})^{\omega}$ and $(\scrG_{f'})^{\omega}$ in such a way that $\calC$ would accept a word in $((\scrG_{f})^{+}(\scrG_{f'})^{+})^{\omega}$, contradicting $((\scrG_{f})^{+}(\scrG_{f'})^{+})^{\omega} \subseteq \scrL(\calS)$.

\begin{definition}[Full Streett Automata]
\label{def:FS}
A family of full Streett automata $\{\calS=\langle\Sigma,S,Q,\Delta,\calF\rangle\}_{n, k >0}$ is such that
\begin{enumerate}[label=\ref{def:FS}.\arabic*,ref=\ref{def:FS}.\arabic*]
\item $S=Q \cup P_{\mathrm{G}} \cup P_{\mathrm{B}} \cup T$ where $Q$, $P_{\mathrm{G}}$, $P_{\mathrm{B}}$ and $T$ are pairwise disjoint sets of the following forms:
\begin{align*}
Q &=\{q_{0},\cdots,q_{n-1}\}, & P_{\mathrm{G}} &=\{g_{1},\cdots,g_{k} \}, &
T &=\{t\}, & P_{\mathrm{B}} &=\{b_{1},\cdots,b_{k} \} \, .
\end{align*}
\item $\calF=\langle G,B\rangle_{I}$ such that $G(i) = \{g_{i}\}$ and $B(i)=\{b_{i}\}$ for $i \in I$.
\end{enumerate}
\end{definition}
%%%%%%%%%%%%%%%%%%%%%%%%%%%%%%%%%%%%%%%%%%%%%%%%%
$Q$ is intended to be the domain of $Q$-rankings. $P_{\mathrm{G}}$ and $P_{\mathrm{B}}$ are pools from which singletons $G(i)$'s and $B(i)$'s are formed. $T$ is to be used for building a \emph{bypass} track that makes graph concatenation behaves like a parallel composition so that properties associated with each subgraph are all preserved in the final concatenation.

\begin{definition}[$Q$-Ranking]
\label{def:Q-ranking} A \emph{$Q$-ranking} for $\calS$ is a function $f: Q \to [1..n] \times I^{k}$, which is identified with a pair of functions $\langle r,h \rangle$, where $r: Q \to [1..n]$ is one-to-one, and $h: Q \to I^{k}$ maps a state to a permutation of $I$.
\end{definition}
For a $Q$-ranking $f=\langle r, h \rangle$, we call $r$ (resp. $h$) the $R$-ranking or numeric ranking (resp. $H$-ranking or index ranking) of $f$. We use \emph{$Q$-ranks} (resp. \emph{$R$-ranks}, \emph{$H$-ranks}) to mean values of $Q$-rankings (resp. $R$-rankings, $H$-rankings). For $q \in Q$, we write $h(q)[i]$ ($i \in I$) to denote the $i$-th component of $h(q)$. Let $\calD^{Q}$ be the set of all $Q$-rankings and $|\calD^{Q}|$ be the size of $\calD^{Q}$. Clearly, we have $n!$ $R$-rankings and $(k!)^{n}$ $H$-rankings, and so $|\calD^{Q}| = (n!)(k!)^{n} = 2^{\Omega(n \lg n + n k \lg k)}$.

As stated in the introduction, $Q$-rankings are essential for obtaining the lower bound. It turns out that $H$-rankings are the core of $Q$-rankings, for $(k!)^{n}$ already begins to dominate $n!$ when $k$ is larger than $\lg n$. Now we explain the idea behind the design of $H$-rankings. Recall that our goal is to have $(\scrG_{f})^{\omega} \not \in \scrL(\calS)$ for any $Q$-ranking $f$ as well as $((\scrG_{f})^{+}(\scrG_{f'})^{+})^{\omega} \subseteq \scrL(\calS)$ for any two different $Q$-rankings $f$ and $f'$. For simplicity, we ignore $R$-rankings and assume $Q$-rankings are just $H$-rankings. We say that a finite path \emph{discharges} obligation $j$ if the path visits $B(j)$ and a finite path \emph{owes} obligation $j$ if the path visits $G(j)$ but does not visit $B(j)$. As shown below, for each $i \in [n]$, $q_{i}$-track in $\scrG_{f}$ is associated with the $k$-tuple $f(q_{i})$, which is a permutation of $I$, and exactly $k$ \emph{full} paths in $\scrG_{f}$ goes from the beginning of $q_{i}$-track to the end of $q_{i}$-track. We say that those paths \emph{on $q_{i}$-track}. For each $i\in [n]$ and $j \in I$, the $j$-th full path on $q_{i}$-track owes exactly the obligation $f(q_{i})[j]$. Let $\varrho=\varrho_{0} \cat \varrho_{1} \cat \cdots$ be an infinite path in $(\scrG_{f})^{\omega}$ where $\varrho_{t}$ ($t \ge 0$) is a full path in the $t$-th $\scrG_{f}$. Without $R$-rankings, our construction prescribes that all $\varrho_{t}$ start and end at a specific track, say $q_{i}$-track, and hence are associated with $f(q_{i})$. Obligations associated with all $\varrho_{t}$ simply form a subset $I'$ of $I$. However, we impose an ordering $\prec_{f, i}$ on $I'$ (different from the standard numeric ordering) such that $f(q_{i})[j] \prec_{f, i} f(q_{i})[j']$ if and only if $j < j'$. The ordering $\prec_{f, i}$ is total thanks to $f(q_{i})$ being a permutation of $I$. Then a condition in our construction guarantees that the minimum obligation with respect to $\prec_{f,i}$ will never be discharged on $\varrho$, and therefore $\varrho$ violates $\langle G,B\rangle_{I}$. Since this $\varrho$ is chosen arbitrarily, we have $(\scrG_{f})^{\omega} \not \in \scrL(\calS)$.

Now let $\scrG \in ((\scrG_{f})^{+}(\scrG_{f'})^{+})^{\omega}$. To show $\scrG \in \scrL(\calS)$, we construct an infinite path $\varrho=\varrho_{0} \cat \varrho_{1} \cat \cdots$ in $\scrG$ that satisfies $\langle G,B\rangle_{I}$, where $\varrho_{t}$ ($t \ge 0$) is a full path in the $t$-th subgraph (which is either $\scrG_{f}$ or $\scrG_{f'}$). Let $i$ be such that $f(q_{i}) \not = f'(q_{i})$ (it is always possible by the assumption $f\not =f'$). Different from before, $q_{i}$-track in $\scrG_{f}$ is associated with $f(q_{i})$ and $q_{i}$-track in $\scrG_{f'}$ is associated with $f'(q_{i})$. Since $f(q_{i})$ and $f'(q_{i})$ are different permutations of $I$, a condition in our construction ensures that a full path $\varrho_{f}$ in $\scrG_{f}$ and a full path $\varrho_{f'}$ in $\scrG_{f'}$, both on $q_{i}$-track, mutually discharge each other's obligations. So we let all $\varrho_{t}$ in $\scrG_{f}$ be $\varrho_{f}$ and all $\varrho_{t}$ in $\scrG_{f'}$ be $\varrho_{f'}$. Since there are infinitely many $\varrho_{f}$ and $\varrho_{f'}$ in $\varrho$, $\varrho$ satisfies $\langle G,B\rangle_{I}$, giving us $\scrG \in \scrL(\calS)$. Since $\scrG$ is chosen arbitrarily, we have $((\scrG_{f})^{+}(\scrG_{f'})^{+})^{\omega} \subseteq \scrL(\calS)$. Now we are read to formally define $Q$-words.

%%%%%%%%%%%%%%%%%%%%%%%%%%%%%%%%%%%%%%%%%%%%%%%
% Definition of $Q$-Word
%%%%%%%%%%%%%%%%%%%%%%%%%%%%%%%%%%%%%%%%%%%%%%%
\begin{definition}[$Q$-Word]
\label{def:Q-word} A finite $\Delta$-graph $\scrG$ is called a \emph{$Q$-word} if every level of $\scrG$ is ranked by the same $Q$-ranking $f=\langle r,h \rangle$ and $\scrG$ satisfies the following additional conditions.
\begin{enumerate}[label=\ref{def:Q-word}.\arabic*,ref=\ref{def:Q-word}.\arabic*]
\item\label{en:Q-word-1} For every $q, q' \in Q$, if $r(q)>r(q')$, there exists a full path $\varrho$ from $\langle q, 0\rangle$ to $\langle q', |\scrG|\rangle$ such that $\varrho$ visits all of $B(1), \ldots, B(k)$.
\item\label{en:Q-word-2} For every $q \in Q$, there exist \emph{exactly} $k$ full paths $\varrho_{1}, \ldots, \varrho_{k}$ from $\langle q, 0\rangle$ to $\langle q, |\scrG|\rangle$ such that for every $i \in I$, $\varrho_{i}$ does not visit $B(h(q)[j])$ for $j \le i$, but visits $B(h(q)[j])$ for $i < j$, and $\varrho_{i}$ does not visit $G(h(q)[j])$ for $j < i$, but visits $G(h(q)[i])$.
\item\label{en:Q-word-3} Only $Q$-vertices have outgoing edges at the first level and incoming edges at the last level.
\item\label{en:Q-word-4} For every $q, q' \in Q$, there exists no full path from $\langle q, 0\rangle$ to $\langle q', |\scrG|\rangle$ if $r(q) < r(q')$.
\end{enumerate}
\end{definition}
%%%%%%%%%%%%%%%%%%%%%%%%%%%%%%%%%%%%%%%%%%%%%%%%
Property~\eqref{en:Q-word-1} concerns with only $R$-rankings. It says that for every two tracks with different $R$-ranks, a path exists that goes from the track with higher rank to the track with the lower rank, and such a path discharges all obligations in $I$. So if those (finite) paths occur infinitely often as fragments of an infinite path $\varrho$, then $\varrho$ clearly satisfies the Streett condition $\langle G,B\rangle_{I}$. Property~\eqref{en:Q-word-2} concerns with only $H$-rankings. It says that exactly $k$ full \emph{``parallel''} paths exist between the two ends of every track, and each owes exactly one distinct obligation in $I$. As shown in Theorem~\ref{thm:lower-bound}, Property~\eqref{en:Q-word-2} is the core of the whole construction and proof, because with $k$ increasing, $H$-rankings contribute more and more to the overall complexity. Properties~\eqref{en:Q-word-3} and~\eqref{en:Q-word-4} are merely technical; they ensure that no other \emph{full paths} exist besides those prescribed by Properties~\eqref{en:Q-word-1} and~\eqref{en:Q-word-2}. Note that in general more than one $Q$-word could exist for a $Q$-ranking $f$. We simply pick an arbitrary one and call it the $Q$-word of $f$, denoted by $\scrG_{f}$.

\begin{theorem}[$Q$-Word]
\label{thm:existence}
A $Q$-word exists for every $Q$-ranking.
\end{theorem}

\begin{example}[$Q$-Word]
\label{ex:Q-word}
Let us consider a full Streett automaton $\calS$ where $n=3$, $k=2$,
\begin{align*}
Q & = \{q_{0}, q_{1}, q_{2} \}, & T & =\{ t \}, &
P_{\mathrm{B}} & = \{b_{1}, b_{2} \}, & P_{\mathrm{G}} &= \{g_{1}, g_{2}\},
\end{align*}
and the following $Q$-ranking $f=\langle r,h \rangle$:
\begin{align*}
r(q_{0}) &= 2, & r(q_{1}) &= 1, & r(q_{2}) &= 3, &
h(q_{0}) &= \langle 1,2 \rangle, & h(q_{1}) &= \langle 1,2 \rangle, & h(q_{2}) &= \langle 2,1 \rangle \, .
\end{align*}
Figure~\ref{fig:Q-word} shows a $Q$-word $\scrG_{f}$, which consists of two subgraphs $\scrG_{r}$ and $\scrG_{h}$, where $\scrG_{r}$ in turn consists of two parts: $\scrG^{(1)}_{r}$ (level $0$ to level $3$) and $\scrG^{(2)}_{r}$ (level $3$ to $6$), and $\scrG_{h}$ in turn consists of three parts: $\scrG^{(0)}_{h}$ (level $6$ to level $12$), $\scrG^{(1)}_{h}$ (level $12$ to level $18$), and $\scrG^{(2)}_{h}$ (level $18$ to level $24$). $\scrG_{r}$ and $\scrG_{h}$ are aimed to satisfy Properties~\eqref{en:Q-word-1} and~\eqref{en:Q-word-2}, respectively.

The $R$-rank (numeric rank) of every level of $\scrG_{r}$ is $(2,1,3)$. In $\scrG^{(1)}_{r}$, a full path $\varrho_{r}$ starts from $\langle q_{2}, 0 \rangle$ whose $R$-rank is the highest. The path visits $\langle b_{1}, 1 \rangle$, $\langle b_{2}, 2 \rangle$ and then $\langle q_{0}, 3 \rangle$ whose $R$-rank is one less than that of $q_{2}$. Similarly in $\scrG^{(2)}_{r}$, the path continues from $\langle q_{2}, 3 \rangle$, visits $\langle b_{1}, 4 \rangle$, $\langle b_{2}, 5 \rangle$ and ends at $\langle q_{1}, 6 \rangle$ whose $R$-rank is one less than that of $q_{0}$.

The $H$-rank (index rank) of every level of $\scrG_{h}$ is $(\langle 1,2\rangle, \langle 1,2\rangle, \langle 2,1\rangle)$. Let us take a look at $\scrG^{(1)}_{h}$. A full path $\varrho_{h}$ (marked green except the last edge) starts at $\langle q_{1}, 12 \rangle$, visits $\langle b_{2}, 13 \rangle$ and $\langle g_{1}, 14 \rangle$ (because of $h(q_{1})[1]=1$), and enters $t$-track (the bypass track $\{t\} \times \bbN$) at $\langle t, 15 \rangle$, from where it stays on $t$-track till reaching $\langle t, 17 \rangle$. Another full path $\varrho'_{h}$ (marked red except the last edge) starts at $\langle q_{1}, 12 \rangle$ too, takes $q_{1}$-track to $\langle q_{1}, 15 \rangle$, and then visits $\langle g_{2}, 16 \rangle$ (because of $h(q_{1})[2]=2$), and enters $t$-track at $\langle t, 17 \rangle$. Both $\varrho_{h}$ and $\varrho'_{h}$ return to $q_{1}$-track at $\langle q_{1}, 18 \rangle$ using the edge $\langle \langle t,17\rangle, \langle q_{1},18\rangle\rangle$ (marked blue). By $\varrho_{0\to6}$, $\varrho_{6\to12}$ and $\varrho_{18\to24}$ (all marked blue) we denote the $q_{1}$-tracks in $\scrG_{r}$, in $\scrG^{(0)}_{h}$ and in $\scrG^{(2)}_{h}$, respectively. It is easy to verify that Property~\eqref{en:Q-word-1} with respect to $q_{2}$ and $q_{1}$ is satisfied by both $\varrho_{r} \cat \varrho_{6\to12} \cat \varrho_{h} \cat \varrho_{18\to24}$ and $\varrho_{r} \cat \varrho_{6\to12} \cat \varrho_{h'} \cat \varrho_{18\to24}$. Also easily seen is that Property~\eqref{en:Q-word-2} with respect to $q_{1}$ is satisfied by $\varrho_{0\to6} \cat \varrho_{6\to12} \cat \varrho_{h} \cat \varrho_{18\to24}$ and $\varrho_{0\to6} \cat \varrho_{6\to12} \cat \varrho_{h'} \cat \varrho_{18\to24}$.
\end{example}

We are ready for the lower bound proof. Let $J \subseteq I$. We use $\langle G,B\rangle_{J}$ to denote the Streett condition with respect to only indices in $J$. The corresponding Rabin condition $[G,B]_{J}$ is similarly defined. When $J$ is a singleton, say $J=\{j\}$, we simply write $\langle G(j),B(j)\rangle$ for $\langle G,B\rangle_{J}$ and $[G(j),B(j)]$ for $[G,B]_{J}$. Obviously, if an infinite run satisfies $\langle G,B\rangle_{J}$ (resp. $[G,B]_{J}$), then the run also satisfies $\langle G,B\rangle_{J'}$ (resp. $[G,B]_{J'}$) for $J' \subseteq J$ (resp. $J \subseteq J' \subseteq I$).

\begin{lemma}
\label{lem:Rabin-nature}
For every $Q$-ranking $f$, $(\scrG_{f})^{\omega} \not \in \scrL(\calS)$.
\end{lemma}

\begin{proof}
Let $f=\langle r,h \rangle$, $\scrG = (\scrG_{f})^{\omega}$ and $\varrho$ an infinite path in $\scrG$. For simplicity, we assume $\varrho$ only lists states appearing on the boundaries of $\scrG_{f}$ fragments; for any $j \ge 0$, $\varrho(j)$ (resp. $\varrho(j+1)$) is a state in the first (resp. last) level of the $j$-th $\scrG_{f}$ fragment. Let $\varrho[j, j+1]$ denote the finite fragment from $\varrho(j)$ to $\varrho(j+1)$. Let $\varrho[j, \infty]$ denote the suffix of $\varrho$ beginning from $\varrho(j)$.

By Property~\eqref{en:Q-word-3}, $\varrho(i) \in Q$ for $i \ge 0$. By Property~\eqref{en:Q-word-4}, $\varrho$ eventually stabilizes on $R$-ranks in the sense that there exists a $j_{0}$ such that for any $j \ge j_{0}$, $r(\varrho(j))=r(\varrho(j+1))$. Because every level of $\scrG$ has the same rank, $\varrho$ stabilizes on a (horizontal) track after $j_{0}$, i.e., there exists $i \in [n]$ such that $\varrho(j) = q_{i}$ for $j \ge j_{0}$. Property~\eqref{en:Q-word-2} says that there are \emph{exactly} $k$ \emph{full} paths $\varrho_{1}, \ldots, \varrho_{k}$ from $\langle q_{i},0\rangle$ to $\langle q_{i}, |\scrG_{f}|\rangle$ in $\scrG_{f}$. Therefore, $\varrho[j_{0}, \infty]$ can be divided into the infinite sequence $\varrho[j_{0}, j_{0}+1], \varrho[j_{0}+1, j_{0}+2], \ldots$, each of which is one of $\varrho_{1}, \ldots, \varrho_{k}$. Let $k_{0} \in I$ be the smallest index such that $\varrho_{k_{0}}$ appears infinitely often in this sequence, i.e., for some $j_{1} \ge j_{0}$, none of $\varrho_{1}, \ldots, \varrho_{k_{0}-1}$ appears in $\varrho[j_{1}, \infty]$. By Property~\eqref{en:Q-word-2} again, $\varrho[j_{1}, \infty]$ visits none of $B(h(q_{i})[1]), \ldots, B(h(q_{i})[k_{0}])$, but visits $G(h(q_{i})[k_{0}])$ infinitely often (because $\varrho_{k_{0}}$ appears infinitely often). In particular, $\varrho$ satisfies $[G(t),B(t)]$ for $t=h(q_{i})[k_{0}]$ and hence $[G,B]_{I}$. Because $\varrho$ is chosen arbitrarily, we have $\scrG \not \in \scrL(\calS)$.
\end{proof}

\begin{lemma}
\label{lem:Streett-potential}
For every two different $Q$-rankings $f$ and $f'$, $((\scrG_{f})^{+} \cat (\scrG_{f'})^{+})^{\omega} \subseteq \scrL(\calS)$.
\end{lemma}

\begin{proof}
Let $\scrG \in ((\scrG_{f})^{+} \cat (\scrG_{f'})^{+})^{\omega}$ be an $\omega$-word where both $\scrG_{f}$ and $\scrG_{f'}$ occur infinitely often in $\scrG$. Let $f = \langle r,h \rangle$ and $f'=\langle r',h' \rangle$. We have two cases: either $r \not = r'$ or $h \not = h'$.

If $r \not = r'$. Since both $r$ and $r'$ are one-to-one functions from $Q$ to $[1..n]$, there must be $i,j \in [n]$ such that $r(q_{i})>r(q_{j})$ and $r'(q_{j})>r'(q_{i})$. By Property~\eqref{en:Q-word-1}, $\scrG_{f}$ contains a full path $\varrho_{i \to j}$ from $\langle q_{i}, 0 \rangle$ to $\langle q_{j}, |\scrG_{f}|\rangle$ that visits all of $B(1), \ldots, B(k)$. By the same property, $\scrG_{f'}$ contains a path $\varrho'_{j \to i}$ from $\langle q_{j}, 0 \rangle$ to $\langle q_{i}, |\scrG_{f'}|\rangle$ that also visits all of $B(1), \ldots, B(k)$. Then $\varrho_{i \to j} \cat \varrho'_{j \to i}$ is a path in $\scrG_{f} \cat \scrG_{f'}$ that visits all of $B(1), \ldots, B(k)$. Also by Property~\eqref{en:Q-word-2}, $\scrG_{f}$ (resp. $\scrG_{f'}$) contains a path $\varrho_{i \to i}$ (resp. $\varrho'_{i \to i}$) from $\langle q_{i}, 0 \rangle$ to $\langle q_{i}, |\scrG_{f}|\rangle$ (resp. from $\langle q_{i}, 0 \rangle$ to $\langle q_{i}, |\scrG_{f'}|\rangle$).

Now we define an infinite path $\hat{\varrho}$ in $\scrG$ as follows. We pick the finite path $\varrho_{i \to i}$ in every $\scrG_{f}$ fragment and $\varrho'_{i \to i}$ in every $\scrG_{f'}$ fragment, except that in the case where a $\scrG_{f}$ fragment is followed immediately by a $\scrG_{f'}$ fragment, we pick $\varrho_{i \to j}$ in the preceding $\scrG_{f}$ and $\varrho'_{j \to i}$ in the following $\scrG_{f'}$. It is easily seen that $\hat{\varrho}$, in the form
\begin{align*}
((\varrho_{i \to i})^{*} \cat (\varrho_{i \to j} \cat \varrho'_{j \to i})^{+} \cat (\varrho'_{i \to i})^{*})^{\omega} \, ,
\end{align*}
visits all of $B(1), \ldots, B(k)$ infinitely often, and hence it satisfies the Streett condition $\langle G,B\rangle_{I}$.

If $h \not = h'$. Then there exist $i \in [n]$, $j \in I$ such that $h(q_{i})[j] \not = h'(q_{i})[j]$ and $h(q_{i})[j^{*}]=h'(q_{i})[j^{*}]$ for $j^{*} \in [1..j-1]$. Since both $h(q_{i})$ and $h'(q_{i})$ are permutations of $I$, we have $j<k$ and
\begin{align}
\label{eq:set-equal}
\{\ h(q_{i})[j^{*}]  \ \mid \ j^{*} \in [j..k] \ \} = \{\ h'(q_{i})[j^{*}] \ \mid \ j^{*} \in [j..k] \ \} \, .
\end{align}
By Property~\eqref{en:Q-word-2}, in $\scrG_{f}$ there exists a path $\varrho_{i \to i}$ from $\langle q_{i}, 0 \rangle$ to $\langle q_{i}, |\scrG_{f}|\rangle$ that visits none of $G(h(q_{i})[j^{*}])$ for $j^{*} \in [1..j-1]$, but visits all of $B(h(q_{i})[j^{*}])$ for $j^{*} \in [j+1..k]$. Similarly, in $\scrG_{f'}$ there exists a path $\varrho'_{i \to i}$ from $\langle q_{i}, 0 \rangle$ to $\langle q_{i}, |\scrG_{f'}|\rangle$ that visits none of $G(h'(q_{i})[j^{*}])$ for $j^{*} \in [1..j-1]$, but visits all of $B(h'(q_{i})[j^{*}])$ for $j^{*} \in [j+1..k]$. Because $h(q_{i})$ and $h'(q_{i})$ are different permutations of $I$, $h'(q_{i})[j]=h(q_{i})[j_{0}]$ for some $j_{0} \in [j+1..k]$ and $h(q_{i})[j]=h'(q_{i})[j_{1}]$ for some $j_{1} \in [j+1..k]$. It follows that both sides of~\eqref{eq:set-equal} are equal to
\begin{align*}
  &\ \ \{\ h(q_{i})[j^{*}]  \ \mid \ j^{*} \in [j+1..k] \ \} \cup \{\ h'(q_{i})[j^{*}] \ \mid \ j^{*} \in [j+1..k] \ \}  \, .
\end{align*}
Therefore $\varrho_{i \to i} \cat \varrho'_{i \to i}$ (in $\scrG_{f} \cat \scrG_{f'}$) visits all of $B(h(q_{i})[j^{*}])$ for $j^{*} \in [j..k]$.

Now let $\hat{\varrho}$ be defined as follows: $\hat{\varrho}$ takes $\varrho_{i \to i}$ in every $\scrG_{f}$ fragment and $\varrho'_{i \to i}$ in every $\scrG_{f'}$ fragment. That is, $\hat{\varrho}$ takes the following form
\begin{align*}
((\varrho_{i \to i})^{+} \cat (\varrho'_{i \to i})^{+})^{\omega} \, .
\end{align*}
Recall that $h(q_{i})[j^{*}]=h'(q_{i})[j^{*}]$ for $j^{*} \in [1..j-1]$. It follows that $\hat{\varrho}$ does not visit any of $G(h(q_{i})[j^{*}])$ for $j^{*} \in [1..j-1]$ because neither $\varrho_{i \to i}$ nor $\varrho'_{i \to i}$ does.  Also since both $\scrG_{f}$ and $\scrG_{f'}$ occur infinitely often in $\scrG$, $\hat{\varrho}$ contains infinitely many $\varrho_{i \to i} \cat \varrho'_{i \to i}$, which implies that $\hat{\varrho}$ visits all of $B(h(q_{i})[j^{*}])$ for $j^{*} \in [j..k]$ infinitely often. Since $h(q_{i})$ is a permutation of $I$, $\hat{\varrho}$ satisfies $\langle G,B\rangle_{I}$.

In either case (whether $r \not = r'$ or $h \not = h'$), $\scrG$ contains a path that satisfies $\langle G,B\rangle_{I}$, which means $\scrG \in \scrL(\calS)$. Because $\scrG$ is arbitrarily chosen, we have $((\scrG_{f})^{+} \cat (\scrG_{f'})^{+})^{\omega} \subseteq \scrL(\calS)$.
\end{proof}

The following lemma is the core of Michel's scheme~\cite{Mic88,Lod99}, recast in the setting of full automata with rankings~\cite{Yan06,CZL09}. Recall that $\calD^{Q}$ denotes the set of all $Q$-rankings and $|\calD^{Q}|$ denotes the cardinality of $\calD^{Q}$.
\begin{lemma}
\label{lem:cost-of-separation}
A union-closed automaton that complements $\calS$ must have at least $|\calD^{Q}|$ states.
\end{lemma}

\begin{proof}
Let $\calC$ be a union-closed automaton that complements $\calS$. By Lemma~\ref{lem:Rabin-nature}, for every $Q$-ranking $f$, $(\scrG_{f})^{\omega} \in \scrL(\calC)$. Let $f$, $f'$ be two different $Q$-rankings and $\scrG_{f}$ and $\scrG_{f'}$ the corresponding $Q$-words. Let $\varrho$ and $\varrho'$ be the corresponding accepting runs of $(\scrG_{f})^{\omega}$ and $(\scrG_{f'})^{\omega}$, respectively. Also let $\varrho_{0}$ and $\varrho_{0}'$, respectively, be the accepting runs of $(\scrG_{f})^{\omega}$ and $(\scrG_{f'})^{\omega}$ when we treat $\scrG_{f}$ and $\scrG_{f'}$ as \emph{atomic} letters, that is, $\varrho_{0}$ (resp. $\varrho'_{0}$) only records states visited at the boundary of $\scrG_{f}$ (resp. $\scrG_{f'}$) and is a subsequence of $\varrho$ (resp. $\varrho'$). Obviously, $\Inf(\varrho_{0}) \subseteq \Inf(\varrho)$, $\Inf(\varrho'_{0}) \subseteq \Inf(\varrho')$, $\Inf(\varrho_{0}) \not = \emptyset$ and $\Inf(\varrho'_{0}) \not = \emptyset$. If $\Inf(\varrho_{0}) \cap \Inf(\varrho'_{0}) = \emptyset$ for any pair of $f$ and $f'$, then clearly $\calC$ has at least $|\calD^{Q}|$ states because the state set of $\calC$ contains $|\calD^{Q}|$ pairwise disjoint nonempty subsets.

Therefore we can assume that $\Inf(\varrho_{0}) \cap \Inf(\varrho'_{0}) \not = \emptyset$ for a fixed pair of $f$ and $f'$. Let $q$ be a state in $\Inf(\varrho_{0}) \cap \Inf(\varrho'_{0})$. Because $q$ occurs infinitely often in $\varrho$, then for some $m > 0$, there exists a path in $(\scrG_{f})^{m}$ that goes from $q$ to $q$ and visits exactly all states in $\Inf(\varrho)$ (or equivalently speaking, $\calC$, upon reading the input word $(\scrG_{f})^{m}$, runs from state $q$ to $q$, visiting exactly all states in $\Inf(\varrho)$ during the run). By $q \xrightarrow[! \Inf(\varrho)]{(\scrG_{f})^{m}} q$ we denote the existence of such a path. Similarly, we have $q \xrightarrow[!\Inf(\varrho')]{(\scrG_{f'})^{m'}} q$ for some $m' > 0$. Also we have $q_{0} \xrightarrow{(\scrG_{f})^{m_{0}}} q$ where $q_{0}$ is an initial state of $\calC$. Now consider the following infinite run $\varrho^{*}$ in the form
\begin{multline*}
q_{0} \xrightarrow{(\scrG_{f})^{m_{0}}} q \xrightarrow[!\Inf(\varrho)]{(\scrG_{f})^{m}} q \xrightarrow[!\Inf(\varrho')]{(\scrG_{f'})^{m'}} q
\xrightarrow[!\Inf(\varrho)]{(\scrG_{f})^{m}} q \xrightarrow[!\Inf(\varrho')]{(\scrG_{f'})^{m'}} q \cdots
\end{multline*}
which is an accepting run of $\calC$ for $(\scrG_{f})^{m_{0}} \cat ((\scrG_{f})^{m} \cat (\scrG_{f'})^{m'})^{\omega}$ because $\Inf(\varrho^{*}) = \Inf(\varrho) \cup \Inf(\varrho')$. However, by Lemma~\ref{lem:Streett-potential}, $(\scrG_{f})^{m_{0}} \cat ((\scrG_{f})^{m} \cat (\scrG_{f'})^{m'})^{\omega} \in ((\scrG_{f})^{+} \cat (\scrG_{f'})^{+})^{\omega} \subseteq \scrL(\calS)$, a contradiction.
\end{proof}

%%%%%%%%%%%%%%%%%%%%%%%%%%%%%%%%%%%%%%%%%%%%%%%%%%%
% Our lower bound
%%%%%%%%%%%%%%%%%%%%%%%%%%%%%%%%%%%%%%%%%%%%%%%%%%%

\begin{theorem}
\label{thm:lower-bound}
Streett complementation is in $2^{\Omega(n\lg n + k n\lg k)}$ for $k=O(n)$ and in $2^{\Omega(n^{2}\lg n)}$ for $k=\omega(n)$, where $n$ and $k$ are the state size and index size of a complementation instance.
\end{theorem}

\begin{proof}
Here we switch to use $n_{0}$ and $k_{0}$, respectively, for the effective state size and index size in our construction $\calS$. We have $n=2k_{0}+n_{0}+1$. By Lemma~\ref{lem:cost-of-separation}, the complementation of $\calS$ requires $|\calD^{Q}|=2^{\Omega(n_{0}\lg n_{0} + n_{0}k_{0}\lg k_{0})}$ states. If $k_{0} \le k$, we can construct a full Streett automaton $\calS'$ with state size $n$ and index size $k$ as follows. $\calS'$ is almost identical to $\calS$ except that its acceptance condition is defined as $\calF'=\langle G',B'\rangle_{I'}$ (for $I'=[1..k]$) such that for $i \in [1..k_{0}]$, $G'(i) = G(i)$ and $B'(i)=B(i)$ and for $i \in [k_{0}+1, k]$, $G'(i)=B'(i)=\emptyset$. It is easily seen that $\calS'$ is equivalent to $\calS$ and hence the complementation lower bound for $\calS$ also applies to that for $\calS'$. Now when $k=O(n)$, we can always find $n_{0}$ and $k_{0}$ such that $k_{0} \le k$, yet $n_{0}=\Omega(n)$ and $k_{0}=\Omega(k)$, and hence we have the lower bound $2^{\Omega(n\lg n + k n\lg k)}$. When $k=\omega(n)$, we set $k_{0}=n_{0}$ so that $k_{0} \le k$, $n_{0}=\Omega(n)$ and $k_{0}=\Omega(n)$, and hence we have the lower bound $2^{\Omega(n^{2}\lg n)}$.
\end{proof} 
\section{Concluding Remarks}
\label{sec:conclusion}
In this paper we proved a tight lower bound $L(n,k)$ for Streett complementation. We note that we can improve the lower bound by two modifications. First, we allow $G(i)$ (resp. $B(i)$) to be arbitrary subsets of $P_{G}$ (resp. $P_{B}$). Second, we also use multi-dimensional $R$-rankings; the range of $r$ is a set of $k$-tuples of integers in $[1..n]$. As a result, both $R$-ranks and $H$-ranks are $k$-tuples of integers where $k$ can be as large as $2^{n}$ (the current effective $k$ is bounded by $n$). These two modifications require much more sophisticated definition of $Q$-rankings and construction of $Q$-words, but they have no asymptotic effect on $L(n,k)$. The situation is different from Rabin complementation~\cite{CZL09}, where $Q$-rankings are also multi-dimensional (though different terms other than $Q$-rankings and $Q$-words were used), and each component in a $k$-tuple (the value of a $Q$-ranking) is independent from one another, and hence each can impose an independent behavior on $Q$-words. Put it in another way, no matter how large the index set is (the maximum size can be $2^{n}$), all dual properties, each of which is parameterized with an index, can be realized in one $Q$-word. For Streett complementation, the diminishing gain when pushing up $k$ made us realize that with increasing number of $Q$-rankings, more and more correlations occur between $Q$-rankings. Exploiting these correlations leads us to the discovery of the corresponding upper bound. 
\begin{sidewaysfigure}
\centering
\begin{align*}
\begin{array}{cc}
\xymatrix@R=0.6 pc@C=1.2pc{
g_{2} &  & \bullet & \bullet & \bullet & \bullet & \bullet & \bullet & \bullet \\
g_{1} &  & \bullet & \bullet & \bullet & \bullet & \bullet & \bullet & \bullet \\
b_{2} &  & \bullet & \bullet & \bullet \ar[ddr] & \bullet & \bullet & \bullet \ar[dddr] & \bullet \\
b_{1} &  & \bullet & \bullet \ar[ur] & \bullet & \bullet & \bullet \ar[ur] & \bullet & \bullet \\
q_{0} & \langle 2,\langle 1,2 \rangle \rangle & \bullet \ar[r] & \bullet \ar[r] & \bullet \ar[r] & \bullet \ar[ur] \ar[r] & \bullet \ar[r] & \bullet \ar[r] & \bullet \\
q_{1} & \langle 1,\langle 1,2 \rangle \rangle & \bullet \ar@{=>}@*{[blue]}@*{[|<0.5pt>]}[r] & \bullet \ar@{=>}@*{[blue]}@*{[|<0.5pt>]}[r] & \bullet \ar@{=>}@*{[blue]}@*{[|<0.5pt>]}[r] & \bullet \ar@{=>}@*{[blue]}@*{[|<0.5pt>]}[r] & \bullet \ar@{=>}@*{[blue]}@*{[|<0.5pt>]}[r] & \bullet \ar@{=>}@*{[blue]}@*{[|<0.5pt>]}[r] & \bullet \\
q_{2} & \langle 3,\langle 2,1 \rangle \rangle & \bullet \ar[uuur] \ar[r] & \bullet \ar[r] & \bullet \ar[r] & \bullet \ar[r] & \bullet \ar[r] & \bullet \ar[r] & \bullet \\
t & & \bullet & \bullet & \bullet & \bullet & \bullet & \bullet & \bullet \\
  & & 00 & 01 & 02 & 03 & 04 & 05 & 06
}  \qquad & \qquad
\xymatrix@R=0.6 pc@C=1.2pc{
g_{2} &  & \bullet & \bullet & \bullet & \bullet & \bullet \ar[dddddddr]& \bullet & \bullet \\
g_{1} &  & \bullet & \bullet & \bullet \ar[ddddddr]& \bullet & \bullet & \bullet & \bullet \\
b_{2} &  & \bullet & \bullet \ar[ur] & \bullet & \bullet & \bullet & \bullet & \bullet \\
b_{1} &  & \bullet & \bullet & \bullet & \bullet & \bullet & \bullet & \bullet \\
q_{0} & \langle 2,\langle 1,2 \rangle \rangle & \bullet \ar[uur] \ar[r] & \bullet \ar[r] & \bullet \ar[r] & \bullet \ar[uuuur] & \bullet & \bullet & \bullet \\
q_{1} & \langle 1,\langle 1,2 \rangle \rangle & \bullet \ar@{=>}@*{[blue]}@*{[|<0.5pt>]}[r] & \bullet \ar@{=>}@*{[blue]}@*{[|<0.5pt>]}[r] & \bullet \ar@{=>}@*{[blue]}@*{[|<0.5pt>]}[r] & \bullet \ar@{=>}@*{[blue]}@*{[|<0.5pt>]}[r] & \bullet \ar@{=>}@*{[blue]}@*{[|<0.5pt>]}[r] & \bullet \ar@{=>}@*{[blue]}@*{[|<0.5pt>]}[r] & \bullet \\
q_{2} & \langle 3,\langle 2,1 \rangle \rangle & \bullet \ar[r] & \bullet \ar[r] & \bullet \ar[r] & \bullet \ar[r] & \bullet \ar[r] & \bullet \ar[r] & \bullet \\
t & & \bullet \ar[r] & \bullet \ar[r] & \bullet \ar[r] & \bullet \ar[r] & \bullet \ar[r] & \bullet \ar[uuur] & \bullet \\
  & & 06 & 07 & 08 & 09 & 10 & 11 & 12
} \\
& \\
(A) \qquad \textit{The $R$-word $\scrG_{r}$} & (B) \qquad \textit{$\scrG^{(0)}_{h}$ of the $H$-word $\scrG_{h}$} \\
& \\
\xymatrix@R=0.6 pc@C=1.2pc{
g_{2} &  & \bullet & \bullet & \bullet & \bullet & \bullet \ar@{=>}@*{[red]}@*{[|<0.5pt>]}[dddddddr] & \bullet & \bullet \\
g_{1} &  & \bullet & \bullet & \bullet \ar@{=>}@*{[green]}@*{[|<0.5pt>]}[ddddddr] & \bullet & \bullet & \bullet & \bullet \\
b_{2} &  & \bullet & \bullet \ar@{=>}@*{[green]}@*{[|<0.5pt>]}[ur] & \bullet & \bullet & \bullet & \bullet & \bullet \\
b_{1} &  & \bullet & \bullet & \bullet & \bullet & \bullet & \bullet & \bullet \\
q_{0} & \langle 2,\langle 1,2 \rangle \rangle & \bullet \ar[r] & \bullet \ar[r] & \bullet \ar[r] & \bullet \ar[r] & \bullet \ar[r] & \bullet \ar[r] & \bullet \\
q_{1} & \langle 1,\langle 1,2 \rangle \rangle & \bullet \ar@{=>}@*{[green]}@*{[|<0.5pt>]}[uuur] \ar@{=>}@*{[red]}@*{[|<0.5pt>]}[r] & \bullet \ar@{=>}@*{[red]}@*{[|<0.5pt>]}[r] & \bullet \ar@{=>}@*{[red]}@*{[|<0.5pt>]}[r] & \bullet \ar@{=>}@*{[red]}@*{[|<0.5pt>]}[uuuuur] & \bullet & \bullet  & \bullet \\
q_{2} & \langle 3,\langle 2,1 \rangle \rangle & \bullet \ar[r] & \bullet \ar[r] & \bullet \ar[r] & \bullet \ar[r] & \bullet \ar[r] & \bullet \ar[r] & \bullet \\
t & & \bullet \ar[r] & \bullet \ar[r] & \bullet \ar[r] & \bullet \ar@{=>}@*{[green]}@*{[|<0.5pt>]}[r] & \bullet \ar@{=>}@*{[green]}@*{[|<0.5pt>]}[r] & \bullet \ar@{=>}@*{[blue]}@*{[|<0.5pt>]}[uur] & \bullet \\
  & & 12
& 13 & 14 & 15 & 16 & 17 & 18
} \qquad & \qquad
\xymatrix@R=0.6 pc@C=1.2pc{
g_{2} &  & \bullet & \bullet & \bullet \ar[dddddddr] & \bullet & \bullet & \bullet & \bullet \\
g_{1} &  & \bullet & \bullet & \bullet & \bullet & \bullet \ar[ddddddr] & \bullet & \bullet \\
b_{2} &  & \bullet & \bullet & \bullet & \bullet & \bullet & \bullet & \bullet \\
b_{1} &  & \bullet & \bullet \ar[uuur] & \bullet & \bullet & \bullet & \bullet & \bullet \\
q_{0} & \langle 2,\langle 1,2 \rangle \rangle & \bullet \ar[r] & \bullet \ar[r] & \bullet \ar[r] & \bullet \ar[r] & \bullet \ar[r] & \bullet \ar[r] & \bullet \\
q_{1} & \langle 1,\langle 1,2 \rangle \rangle & \bullet \ar@{=>}@*{[blue]}@*{[|<0.5pt>]}[r] & \bullet \ar@{=>}@*{[blue]}@*{[|<0.5pt>]}[r] & \bullet \ar@{=>}@*{[blue]}@*{[|<0.5pt>]}[r] & \bullet \ar@{=>}@*{[blue]}@*{[|<0.5pt>]}[r] & \bullet \ar@{=>}@*{[blue]}@*{[|<0.5pt>]}[r] & \bullet \ar@{=>}@*{[blue]}@*{[|<0.5pt>]}[r] & \bullet \\
q_{2} & \langle 3,\langle 2,1 \rangle \rangle & \bullet \ar[uuur] \ar[r] & \bullet \ar[r] & \bullet \ar[r] & \bullet \ar[uuuuur] & \bullet & \bullet & \bullet \\
t & & \bullet \ar[r] & \bullet \ar[r] &
\bullet \ar[r] & \bullet \ar[r] & \bullet \ar[r] & \bullet \ar[ur] & \bullet \\
  & & 18 & 19 &
20 & 21 & 22 & 23 & 24
} \\
& \\
(C) \qquad \textit{$\scrG^{(1)}_{h}$ of the $H$-word $\scrG_{h}$} & (D) \qquad \textit{$\scrG^{(2)}_{h}$ of the $H$-word $\scrG_{h}$}
\end{array}
\end{align*}
\caption{$Q$-word $\scrG_{f}$ ($f=\langle r,h\rangle$)}
\label{fig:Q-word}
\end{sidewaysfigure}

\newpage
\appendix
\section{Proofs}
\label{sec:proofs}

In this section we prove Theorem~\ref{thm:existence}. Recall that we need a construction to simultaneously satisfy all properties in Definition~\ref{def:Q-word}, which are parameterized with pairs of states (Condition~\eqref{en:Q-word-1}) or states (Condition~\eqref{en:Q-word-2}). The idea is to concatenate a sequence of finite $\Delta$-graphs, each of which satisfies the properties with respect to a specific pair of states or a specific state. With the help of the bypass track ($t$-track), properties associated with each individual subgraph are all preserved in the final concatenation, giving us a desired $Q$-word.

Let $f=\langle r,h \rangle$. $\scrG_{f}$ divides into two sequential subgraphs $\scrG_{r}$ and $\scrG_{h}$, which satisfy Properties~\eqref{en:Q-word-1} and~\eqref{en:Q-word-2}, respectively. Properties~\eqref{en:Q-word-3} and~\eqref{en:Q-word-4} are obvious once the final construction is shown. As stated earlier, Property~\eqref{en:Q-word-1} and Property~\eqref{en:Q-word-2} are orthogonal; Property~\eqref{en:Q-word-1} only relies on $R$-rankings and Property~\eqref{en:Q-word-2} only relies on $H$-rankings. We call a finite $\Delta$-graph whose every level is ranked by the same $R$-ranking, an \emph{$R$-word} if it satisfies Properties~\eqref{en:Q-word-1},~\eqref{en:Q-word-3} and~\eqref{en:Q-word-4}. Similarly, a finite $\Delta$-graph whose every level is ranked by the same $H$-ranking, is called an \emph{$H$-word} if it satisfies Properties~\eqref{en:Q-word-2},~\eqref{en:Q-word-3} and~\eqref{en:Q-word-4}. As $Q$-words, $H$-words (resp. $R$-words) are not uniquely determined by $H$-rankings (resp. $R$-rankings). Nevertheless, all $H$-words (resp. $R$-words) corresponding to a specific $h$ (resp. $r$) serve the construction purpose equally well, and hence we simply name an arbitrarily chosen one by $\scrG_{h}$ (resp. $\scrG_{r}$). Theorem~\ref{thm:existence} builds on Lemmas~\ref{lem:R-word} and~\ref{lem:H-word}.

\begin{lemma}[$R$-Word]
\label{lem:R-word}
An $R$-word exists for every $R$-ranking.
\end{lemma}

\begin{proof}
Let $r$ be an $R$-ranking. $\scrG_{r}$ is constructed as follows. We order $Q$ as $q_{m_{1}}, \ldots, q_{m_{n}}$ such that $r(q_{m_{1}}) > \cdots > r(q_{m_{n}})$. $\scrG_{r}$ has $n-1$ parts $\scrG^{(1)}_{r}, \ldots, \scrG^{(n-1)}_{r}$. In $\scrG^{(i)}_{r}$ ($i \in [1..n-1]$), a path leaves $q_{m_{i}}$ whose $R$-rank is the $i$-th large, visits all $B(j)$-vertices ($j \in I$) and ends at $q_{m_{i+1}}$ whose $R$-rank is the $(i+1)$-th large. Formally we define the following letters
\begin{align*}
Id(Q) &= \{ \, \langle q_{i}, q_{i} \rangle \ \mid \ i \in [n] \, \},  && \displaybreak[0]\\
Q(i)ToB(1) &= Id(Q)  \cup \{\, \langle q_{i}, b_{1} \rangle \, \}, && (i \in [n]) \displaybreak[0]\\
B(i)ToB(i+1) &= Id(Q)  \cup \{\, \langle b_{i}, b_{i+1} \rangle\, \}, && (i \in [1..k-2]) \displaybreak[0]\\
B(k)ToQ(i) &= Id(Q)  \cup \{\, \langle b_{k}, q_{i} \rangle\, \}, && (i \in [n]) \displaybreak[0]
\end{align*}
and then define $\scrG_{r}$ as
\begin{multline*}
Q(m_{1})ToB(1) \cat B(1)ToB(2) \cat \cdots \cat B(k)ToQ(m_{2}) \\
\cat \cdots \cat
Q(m_{n-1})ToB(1) \cat B(1)ToB(2) \cat \cdots \cat B(k)ToQ(m_{n}) \, .
\end{multline*}
We verify that $\scrG_{r}$ satisfies Property~\eqref{en:Q-word-1}. Let $q, q' \in Q$ be such that $r(q) > r(q')$. Let $i, i' \in [1..n]$ be such that $i < i'$, $q=q_{m_{i}}$ and $q'=q_{m_{i'}}$. Recall that by a \emph{full path} in $\scrG$ we mean a path going from level $0$ to level $|\scrG|$. We define a full path $\varrho_{i, i'}$ in $\scrG_{r}$ as follows. The path $\varrho_{i, i'}$ takes $q_{m_{i}}$-track until it reaches the left boundary of the letter $Q(m_{i})ToB(1)$, from where it leaves $q_{m_{i}}$-track to visit $b_{1}, \ldots, b_{k}$ (in this order) and then $q_{m_{i+1}}$. Continuing from $q_{m_{i+1}}$, $\varrho_{i, i'}$ follows the same pattern till it reaches $q_{m_{i+2}}$. Repeating this pattern $i'-i$ times, $\varrho_{i, i'}$ reaches $q_{m_{i'}}$ from where it takes $q_{m_{i'}}$-track till the end of $\scrG_{r}$. In summary, $\varrho_{i, i'}$ takes the form
\begin{multline*}
q=q_{m_{i}} \to \cdots \to q_{m_{i}} \displaybreak[0]
\to b_{1} \to \cdots \to b_{k} \to q_{m_{i+1}} \displaybreak[0]
\to b_{1} \to \cdots \to b_{k} \to q_{m_{i+2}} \displaybreak[0] \\
\to \ \cdots \cdots \cdots \ \to b_{1} \to \cdots \to b_{k} \to q_{m_{i'}} \to \cdots \to q_{m_{i'}}=q' \, .
\end{multline*}
Easily seen from the construction, with respect to any pair $q$ and $q'$ where $r(q) > r(q')$, Property~\eqref{en:Q-word-1} is satisfied by the corresponding $\varrho_{i, i'}$. Properties~\eqref{en:Q-word-3} and~\eqref{en:Q-word-4} are immediate from the construction.
\end{proof}

\begin{example}[$R$-Word]
\label{ex:R-word}
Let us revisit Example~\ref{ex:Q-word}. $Q$ is ordered as $q_{2}, q_{0}, q_{1}$ for $r(q_{2}) > r(q_{0}) > r(q_{1})$. So $m_{1}=2$, $m_{2}=0$ and $m_{3}=1$. In Figure~\ref{fig:Q-word}, the $R$-word $\scrG_{r}$ consists of two parts: $\scrG^{(1)}_{r}$ (level $0$ to level $3$) and $\scrG^{(2)}_{r}$ (level $3$ to $6$), defined as follows.
\begin{align*}
\scrG^{(1)}_{r} &: Q(2)ToB(1) \cat B(1)ToB(2) \cat B(2)ToQ(0), &
\scrG^{(2)}_{r} &: Q(0)ToB(1) \cat B(1)ToB(2) \cat B(2)ToQ(1).
\end{align*}
For Property~\eqref{en:Q-word-1} with respect to $q_{2}$ and $q_{1}$, we can obtain the desired $\varrho_{2, 1}$ as follows. In $\scrG^{(1)}_{r}$, $\varrho_{2, 1}$ starts from $\langle q_{2}, 0 \rangle$, visits $\langle b_{1}, 1 \rangle$, $\langle b_{2}, 2 \rangle$ and then $\langle q_{0}, 3 \rangle$. In $\scrG^{(2)}_{r}$, $\varrho_{2, 1}$ continues from $\langle q_{0}, 3 \rangle$, visits $\langle b_{1}, 4 \rangle$, $\langle b_{2}, 5 \rangle$ and lands at $\langle q_{1}, 6 \rangle$. For Property~\eqref{en:Q-word-1} with respect to $q_{0}$ and $q_{1}$, we can obtain the desired $\varrho_{0, 1}$ as follows. In $\scrG^{(1)}_{r}$, $\varrho_{0, 1}$ starts from $\langle q_{0}, 0 \rangle$, passes through $\scrG^{(1)}_{r}$ via $q_{0}$-track until it reaches $\langle q_{0}, 3 \rangle$ from where it visits $\langle b_{1}, 4 \rangle$, $\langle b_{2}, 5 \rangle$ and lands at $\langle q_{1}, 6 \rangle$.
\end{example}

\begin{lemma}[$H$-Word]
\label{lem:H-word}
An $H$-word exists for every $H$-ranking.
\end{lemma}

\begin{proof}
Let $h$ be an $H$-ranking. $\scrG_{h}$ is constructed as follows. $\scrG_{h}$ comprises $n$ sequential parts $\scrG^{(0)}_{h}, \ldots \scrG^{(n-1)}_{h}$, and for each $i \in [n]$, $\scrG^{(i)}_{h}$ in turn comprises $k$ sequential parts $\scrG^{(i,1)}_{h}, \ldots, \scrG^{(i,k)}_{h}$. To fulfill the requirement with respect to a pair $q_{i} \in Q$  ($i \in [n]$) and $j \in I$ in Property~\eqref{en:Q-word-2}, we select a full path $\varrho_{i,j}$ in $\scrG_{h}$ as follows. The path starts from $\langle q_{i}, 0 \rangle$ and ends at $\langle q_{i}, |\scrG_{h}|\rangle$. The path $\varrho_{i,j}$ simply passes through, via $q_{i}$-track, all $\scrG^{(i')}_{h}$ for $i' \not = i$. In $\scrG^{(i)}_{h}$, $\varrho_{i,j}$ also passes through $\scrG^{(i,1)}_{h}, \ldots, \scrG^{(i,j-1)}_{h}$ via $q_{i}$ track until it reaches the beginning of $\scrG^{(i,j)}_{h}$, from where it visits $B(h(q_{i})[j+1]), \ldots, B(h(q_{i})[k]), G(h(q_{i})[j])$ (in this order), and then enters $t$-track. The path continues and stays on $t$-track till arriving at the second last level of $\scrG^{(i)}_{h}$, and then ending at $\langle q_{i}, |\scrG^{(i)}_{h}|\rangle$. Formally we define the following letters.
\begin{align*}
Id(Q) &= \{\, \langle q, q \rangle \, \mid \, q \in Q \}, && \displaybreak[0] \\
Id(T) &= \{\, \langle t, t \rangle \}, && \displaybreak[0] \\
Q(i)toB(j) &= Id(Q) \cup Id(T) \cup \{\, \langle q_{i},b_{j}\rangle\, \}, && (i \in [n], j \in I) \displaybreak[0] \\
B(i)ToB(j) &= Id(Q) \cup Id(T) \cup \{\, \langle b_{i},b_{j}\rangle\, \}, && (i,j \in I) \displaybreak[0] \\
B(i)ToG(j) &= Id(Q) \cup Id(T) \cup \{\, \langle b_{i},g_{j}\rangle\, \}, && (i,j \in I) \displaybreak[0] \\
Q(i)ToG(j) &= Id(Q) \cup Id(T) \cup \{\, \langle q_{i},g_{j}\rangle\, \}, && (i \in [n], j \in I) \displaybreak[0] \\
G(i)ToT &= Id(Q) \cup Id(T) \cup \{\, \langle g_{i}, t \rangle\, \}, && (i \in I) \displaybreak[0] \\
Q(i)To^{-}G(j) &= Id(Q) \cup Id(T)
\cup \{\, \langle q_{i},g_{j}\rangle\, \} \setminus \{\, \langle q_{i},q_{i}\rangle\, \},
&& (i \in [n], j \in I) \displaybreak[0] \\
G(i)To^{-}T &= Id(Q) \cup Id(T)
\cup \{\, \langle g_{i}, t \rangle\, \} \setminus \{\, \langle q_{i},q_{i}\rangle\, \},
&& (i \in I) \displaybreak[0] \\
TTo^{-}Q(i) &= Id(Q) \cup \{\, \langle t, q_{i} \rangle\, \} \setminus \{\, \langle q_{i},q_{i}\rangle\, \}. && (i \in [n]) \displaybreak[0] \\
\end{align*}
Note that letters of the forms $Q(i)To^{-}G(j)$ or $G(i)To^{-}T$ do not contain horizontal edges $\langle q_{i},q_{i}\rangle$. These letters are used in $\scrG^{(i,k)}_{h}$ so that a full path in $\scrG^{(i)}_{h}$ from $\langle q_{i},0\rangle$ to $\langle q_{i}, |\scrG^{(i)}_{h}|\rangle$ has to leave $q_{i}$-track first and end up at $t$-track. Letters $TTo^{-}Q(i)$ contain neither the bypass edges $\langle t,t\rangle$ nor horizontal edges $\langle q_{i},q_{i}\rangle$. These letters are also used in $\scrG^{(i,k)}_{h}$ so that all full paths in $\scrG^{(i)}_{h}$ using $t$-track end at $\langle q_{i}, |\scrG^{(i)}_{h}|\rangle$. Formally, $\scrG_{h}=\scrG^{(0)}_{h} \cat \scrG^{(2)}_{h} \cat \cdots \cat \scrG^{(n-1)}_{h}$, and for each $i \in [n]$, $\scrG^{(i)}_{h}=\scrG^{(i,1)}_{h} \cat \scrG^{(i,2)}_{h} \cat \cdots \cat \scrG^{(i,k)}_{h}$, where for each $j \in [1..k-1]$, $\scrG^{(i,j)}_{h}$ is
\begin{multline}
Q(i)ToB(h(q_{i})[j+1]) \cat B(h(q_{i})[j+1])ToB(h(q_{i})[j+2]) \\
\cat \cdots \cat B(h(q_{i})[k-1])ToB(h(q_{i})[k]) \cat B(h(q_{i})[k])ToG(h(q_{i})[j]) \cat G(h(q_{i})[j])ToT \, ,
\label{eq:rho_i_j}
\end{multline}
and finally $\scrG^{(i,k)}_{h}$ is
\begin{align}
Q(i)To^{-}G(h(q_{i})[k]) \cat G(h(q_{i})[k])To^{-}T \cat TTo^{-}Q(i) . \label{eq:rho_i_k}
\end{align}

We verify that Property~\eqref{en:Q-word-2} holds for every pair $q_{i} \in Q$ ($i \in [n]$) and $j \in I$. First consider $j \in [1..k-1]$. By~\eqref{eq:rho_i_j}, in $\scrG^{(i,j)}_{h}$ a full path $\varrho'_{i,j}$ exists that starts from $\langle q_{i}, 0 \rangle$, visits $B(h(q_{i})[j+1]), \ldots, B(h(q_{i})[k])$ and $G(h(q_{i})[j])$, and finally ends at $\langle t, |\scrG^{(i,j)}_{h}|\rangle$. We extend $\varrho'_{i,j}$ to a full path $\varrho_{i,j}$ in $\scrG_{h}$ as follows. The path $\varrho_{i,j}$ takes $q_{i}$-track in all $\scrG^{(i')}_{h}$ for $i' \not = i$. Inside $\scrG^{(i)}_{h}$, $\varrho_{i,j}$ also takes $q_{i}$-track in all $\scrG^{(i,j')}_{h}$ for $j' < j$. Inside $\scrG^{(i,j)}_{h}$, $\varrho_{i,j}$ is just $\varrho'_{i,j}$. In $\scrG^{(i,j')}_{h}$ for $j' > j$, $\varrho_{i,j}$ takes $t$-track till it reaches the second last level of $\scrG^{(i,k)}_{h}$, from where it takes the edge $\langle t, q_{i} \rangle$ to the $q_{i}$ at the end of $\scrG^{(i)}_{h}$. Put all together, for every $i \in [n], j \in [1..k-1]$, $\varrho_{i,j}$ takes the form
\begin{multline*}
q_{i} \to \cdots \to q_{i} \to b_{h(q_{i})[j+1]} \to \cdots \to b_{h(q_{i})[k]} \to g_{h(q_{i})[j]}
\to t \to \cdots \to t \to q_{i} \to \cdots \to q_{i} \, .
\end{multline*}

The case $j = k$ is similar. By~\eqref{eq:rho_i_k}, in $\scrG^{(i,k)}_{h}$ a full path $\varrho'_{i,k}$ exists that starts from $\langle q_{i}, 0 \rangle$, visits $G(h(q_{i})[k])$, arrives at the $t$ at the second last level of $\scrG^{(i,k)}_{h}$, and finally takes the edge $\langle t, q_{i} \rangle$ back to the $q_{i}$ at the end of $\scrG^{(i,k)}_{h}$ (also at the end of $\scrG^{(i)}_{h}$). We extend $\varrho'_{i,k}$ to a full path $\varrho_{i,k}$ in $\scrG_{h}$ in the same way as before. The only difference is that $\varrho_{i,k}$ simply passes through, via $q_{i}$-track, all $\scrG^{(i',j')}_{h}$ for any $i' \not = i$ or $j' \not = k$. Put all together, for every $i \in [n]$, $\varrho_{i,k}$ takes the form
\begin{align*}
q_{i} \to \cdots \to q_{i} \to g_{h(q_{i})[k]} \to t \to q_{i} \to \cdots \to q_{i} \, .
\end{align*}
Note that for any $i \in [n]$, $j \in I$, the path $\varrho_{i,j}$ has to leave $q_{i}$-track in $\scrG^{(i,j)}_{h}$ to fulfill the requirement with respect to $q_{i}$ and $j$, and it has to use $t$-track and the edge $\langle t, q_{i} \rangle$ in $\scrG^{(i,k)}_{h}$ to return to $q_{i}$-track at the end of $\scrG^{(i)}_{h}$, because the $q_{i}$-track in $\scrG^{(i,k)}_{h}$ is broken at the vertex from where $\varrho_{i,k}$ starts to fulfill the requirement with respect to $q_{i}$ and $k$ in Property~\eqref{en:Q-word-2}.

We are done with the existence part of Property~\eqref{en:Q-word-2}. As of the exactness part, we note that the following facts hold for every $i\in [n]$.
\begin{enumerate}
\item Vertices $q_{i}$ have only a horizontal outgoing edge in all $\scrG^{(i')}_{h}$ for $i' \not= i$.
\item Each $\scrG^{(i,j)}$ ($j \in I$) contains exactly one $q_{i}$-vertex that has exactly one \emph{non-horizontal} outgoing edge. So in $\scrG^{(i)}_{h}$, there are $k$ such $q_{i}$-vertices in total.
\item The $q_{i}$-track is broken in $\scrG^{(i)}_{h}$ (more precisely, at the beginning of $\scrG^{(i,k)}_{h}$).
\item Any full path in $\scrG_{h}$ from $\langle q_{i}, 0\rangle$ to $\langle q_{i}, |\scrG_{h}|\rangle$ has to take one of non-horizontal outgoing edges in $\scrG^{(i)}_{h}$.
\item If a path in $\scrG^{(i)}_{h}$ takes a non-horizontal edge to leave $q_{i}$-track, then the path has to land on $t$-track and stay on $t$-track till returning to the $q_{i}$ at the end of $\scrG^{(i)}_{h}$.
\end{enumerate}
The exactness part then follows. Property~\eqref{en:Q-word-3} is immediate as before. Property~\eqref{en:Q-word-4} holds due to the fact that for every $i, j \in [n]$, any full path in $\scrG^{(i)}_{h}$ that starts from $\langle q_{j},0\rangle$ ends at $\langle q_{j},|\scrG^{(i)}_{h}|\rangle$.
\end{proof}

\begin{example}[$H$-Word]
\label{ex:H-word}
Let us revisit Example~\ref{ex:Q-word}. In Figure~\ref{fig:Q-word}, the $H$-word $\scrG_{h}$ consists of three parts: $\scrG^{(1)}_{h}$ (level $6$ to level $12$), $\scrG^{(2)}_{h}$ (level $12$ to level $18$), and $\scrG^{(3)}_{h}$ (level $18$ to level $24$), defined as follows: $\scrG^{(i)}_{h}=\scrG^{(i,1)}_{h} \cat \scrG^{(i,2)}_{h}$ for $i \in [3]$ and
\begin{align*}
\scrG^{(0,1)}_{h} &= Q(0)ToB(2) \cat B(2)ToG(1) \cat G(1)ToT , &
\scrG^{(0,2)}_{h} &= Q(0)To^{-}G(2) \cat G(2)To^{-}T \cat TTo^{-}Q(0) \, ,\displaybreak[0]\\
\scrG^{(1,1)}_{h} &= Q(1)ToB(2) \cat B(2)ToG(1) \cat G(1)ToT , &
\scrG^{(1,2)}_{h} &= Q(1)To^{-}G(2) \cat G(2)To^{-}T \cat TTo^{-}Q(1) \, ,\displaybreak[0]\\
\scrG^{(2,1)}_{h} &= Q(2)ToB(1) \cat B(1)ToG(2) \cat G(2)ToT , &
\scrG^{(2,2)}_{h} &= Q(2)To^{-}G(1) \cat G(1)To^{-}T \cat TTo^{-}Q(2) \, . \displaybreak[0]
\end{align*}
Let us take a look the paths $\varrho_{h}$ and $\varrho_{h'}$ (in $\scrG^{(1)}_{h}$) defined in Example~\ref{ex:Q-word}. The path $\varrho_{h}$ (marked green except the last edge) starts at $\langle q_{1}, 12 \rangle$, visits $\langle b_{2}, 13 \rangle$ and $\langle g_{1}, 14 \rangle$, and enters $t$-track at $\langle t, 15 \rangle$. It continues on $t$-track till reaching $\langle t, 17 \rangle$, and then takes $\langle \langle t, 17 \rangle, \langle q_{1}, 18 \rangle \rangle$ (marked blue) to the end. The path $\varrho_{h'}$ (marked red except the last edge) starts at $\langle q_{1}, 12 \rangle$, takes $q_{1}$-track to reach $\langle q_{1}, 15 \rangle$, from where it visits $\langle g_{2}, 16 \rangle$ and then enters $t$-track at $\langle t, 17 \rangle$. Same as $\varrho_{h}$, $\varrho_{h'}$ returns to $q_{1}$-track via $\langle \langle t, 17 \rangle, \langle q_{1}, 18 \rangle \rangle$.
\end{example}

\begin{varthm}{\ref{thm:existence}~[$Q$-Words].}
A $Q$-word exists for every $Q$-ranking.
\end{varthm}
\begin{proof}
By Lemmas~\ref{lem:R-word} and~\ref{lem:H-word}, we have $\scrG_{r}$ and $\scrG_{h}$ as an $R$-word and an $H$-word, respectively. The desired $Q$-word $\scrG$ is just $\scrG_{r} \cat \scrG_{h}$. Properties~\eqref{en:Q-word-3} and~\eqref{en:Q-word-4} follow immediately because they hold both in $\scrG_{r}$ and $\scrG_{h}$.

Let $\varrho^{r}_{i, i'}$ be the full path in $\scrG_{r}$ that satisfies Property~\eqref{en:Q-word-1} for $q_{i}$ and $q_{i'}$ where $i, i' \in [n]$ and $r(q_{i}) > r(q_{i'})$, and $\varrho^{h}_{i', k}$ the full path in $\scrG_{h}$ that satisfies Property~\eqref{en:Q-word-2} (with respect to $q_{i'}$ and index $k$). Then $\varrho^{r}_{i, i'} \cat  \varrho^{h}_{i', k}$ is the path that Property~\eqref{en:Q-word-1} requires for vertex pair $q_{i}$ and $q_{i'}$.

Let $i \in [n]$ and $j \in I$. Let $\varrho^{r}_{i,i}$ be the full $q_{i}$-track in $\scrG_{r}$, and $\varrho^{h}_{i, j}$ the full path in $\scrG_{h}$ that satisfies Property~\eqref{en:Q-word-2} (with respect to vertex $q_{i}$ and index $j$). Then in $\scrG$, for each $q \in Q$, we have $k$ full paths $\varrho^{r}_{i, i} \cat \varrho^{h}_{i, j}$ ($j \in I$), which takes care of the existence part of Property~\eqref{en:Q-word-2}. The exactness part follows from the exactness part of Property~\eqref{en:Q-word-2} for $\scrG_{h}$, and the fact that for each $i \in [n]$, $\varrho^{r}_{i,i}$ is unique in $\scrG_{r}$.
\end{proof}

\end{document}